\documentclass[10pt,conference]{IEEEtran}

\usepackage{amsthm,amsmath,amsfonts,braket,algorithm,graphicx, braket,balance}

\theoremstyle{definition} \newtheorem{define} {Definition} [section]
\theoremstyle{definition} 
\newtheorem {theorem} {Theorem}
\newtheorem {corollary} {Corollary}[section]
\newtheorem {lemma} {Lemma}

\newcommand{\bk}[1]{\braket{#1|#1}}

\newcommand{\X}{\mathcal{X}_\alpha}
\newcommand{\num}{\#}

\newcommand{\kb}[1]{\mathbf{\left[#1\right]}}

\newcommand{\al}{\mathcal{A}}

\newcommand{\trd}[1]{\left|\left| #1 \right| \right|}

\newcommand{\st}{\text{ } : \text{ }}

\newcommand{\dc}{\sim_\delta}

\newcommand{\Hmin}{H_\infty}
\newcommand{\Hmax}{H_{max}}

\newcommand{\hd}{\Delta_H}
\newcommand{\uniform}{\mathcal{U}}

\newcommand{\leakEC}{\lambda_{EC}}

\floatname{algorithm}{Protocol}

\begin{document}
\title{Entropic Uncertainty for Biased Measurements}

\author{
\IEEEauthorblockN{Walter O. Krawec}
\IEEEauthorblockA{
\textit{University of Connecticut}\\
Storrs CT, USA\\
walter.krawec@uconn.edu
}
}

\maketitle
\begin{abstract}
Entropic uncertainty relations are powerful tools, especially in quantum cryptography.  They typically bound the amount of uncertainty a third-party adversary may hold on a measurement outcome as a result of the measurement overlap.  However, when the two measurement bases are biased towards one another, standard entropic uncertainty relations do not always provide optimal lower bounds on the entropy.  Here, we derive a new entropic uncertainty relation, for certain quantum states and for instances where the two measurement bases are no longer mutually unbiased.  We evaluate our bound on two different quantum cryptographic protocols, including BB84 with faulty/biased measurement devices, and show that our new bound can produce higher key-rates under several scenarios when compared with prior work using standard entropic uncertainty relations.
\end{abstract}

\section{Introduction}

Quantum entropic uncertainty relations are a powerful tool in quantum information theory and quantum cryptography.  Such relations typically bound the amount of uncertainty in the outcome of two different measurements as a function only of the measurements themselves.  For instance, the famous Maassen and Uffink inequality \cite{maassen1988generalized} states that if a quantum state is measured in one of two bases $Z$ or $X$, then $H(Z) + H(X) \ge c$, where $H(Z)$ is the entropy in the $Z$ basis outcome (similar for $H(X)$), and $c$ is a function of the ``measurement overlaps'' between the $X$ and $Z$ bases and is maximal whenever $Z$ and $X$ are mutually unbiased bases.  By now there are a large variety of different entropic uncertainty relations \cite{tomamichel2011uncertainty,bialynicki2006formulation,berta2010uncertainty,adabi2016tightening}; see \cite{coles2017entropic} for a general survey.

One very useful entropic uncertainty relation was introduced in \cite{tomamichel2011uncertainty} which bounds the quantum min entropy - a quantity we define formally later, but denote by $\Hmin(A|E)$.  Min entropy is a very useful resource to measure as it is directly related to how many uniform random secret bits may be extracted from a quantum state \cite{renner2008security}.  In a little detail, let's assume $\rho_{ABE}$ is a quantum state where the $A$ and $B$ registers consist of $n$ qubits each and let $Z = \{\ket{0}, \ket{1}\}$ be the standard computational basis for qubits and $X = \{\ket{x_0}, \ket{x_1}\}$ be some other basis with $\ket{x_0} = \sqrt{1/2 + b}$ and $\ket{x_1} = \sqrt{1/2-b}$ for some ``bias'' parameter $b \in [0, .5]$ (e.g., this may be the Hadamard basis if $b=0$).  Note the results will be symmetric if we have $b \in [-.5, 0]$.  Assume a measurement is made on the $A$ system in either the $Z$ basis (resulting in some random variable $A_Z$) or the $X$ basis (yielding random variable $A_X$); similar for the $B$ system. Then, the relation defined in \cite{tomamichel2011uncertainty} roughly states (when restricted to basis measurements of this form), that:
\begin{equation}\label{eq:std-eu}
  \Hmin(A_Z|E) + \Hmax(A_X|B_X) \ge -n\cdot\log_2\left(\frac{1}{2}+b\right).
\end{equation}
Note that the lower-bound is maximal when $b=0$ and one gets $\Hmin(A_Z|E) + \Hmax(A_X|B_X) \ge n$.  This relation is used many times in various quantum cryptographic proofs of security as it allows one to bound the quantum min entropy between Alice and an adversary system Eve, simply as a function of the measurements performed and $\Hmax(A_X|B_X)$, the latter of which may be easily bounded through standard classical sampling arguments and is generally a function of the ``error'' induced in the quantum communication line.

The above expression, as stated, is not only highly useful, but also widely applied.  However, when $b \ne 0$, it is not difficult to see that the lower bound on $\Hmin(A|E)$ begins to drop rapidly.  In this work, we derive a new entropic uncertainty relation for cases when there is non-zero bias in the measurement bases.  Our new relation, though stated formally in Theorem \ref{thm:main}, roughly takes the form:
\begin{equation}\label{eq:new-eu}
  \Hmin(A_Z|E) + n\cdot h\left(Q_X+4b^2 + \epsilon\right) \ge n,
\end{equation}
where $h(x)$ is the binary entropy, $Q_X$ is the relative number of errors in Alice and Bob's $X$ basis measurement, and $\epsilon$ is a function of the number of qubits that were measured in the $X$ basis (and which goes to zero in the asymptotic limit).  Note, the above is only true if $Q_X + 4b^2 + \epsilon < 1/2$ which can be checked by the users of the protocol before continuing.  This already puts an upper-bound on $b$ of $\sqrt{1/8} \approx 0.3535$ (unlike Equation \ref{eq:std-eu} which has an upper bound of $b < 1/2$).  Thus, when there is bias but no noise ($Q_X = 0$ and $\Hmax(A_X|B_X) = 0$), our result performs worse; however, importantly, when there is both noise and bias, our bound often outperforms Equation \ref{eq:std-eu}, sometimes substantially so as our later evaluations show.  Thus, it can be immediately applied to cryptographic proofs of protocols where measurements are biased and there is noise in the channel (either natural noise or adversarial noise) and used to show that higher bit generation rates are possible under these circumstances.  \emph{We comment that our proof in this paper requires} a particular (though arguably minimal, and even enforceable by the users, as we comment later) assumption on the quantum state under investigation.  However, this assumption is only needed in one part of the proof and we suspect our methods can be suitably extended to work, with the same result, even without this assumption.  However, this we leave as future work.

Our relation is a so-called \emph{sampling-based entropic uncertainty relation}, which is a class of entropic uncertainty relations introduced in \cite{krawec2019quantum,yao2022quantum}. These relations utilize a quantum sampling framework of Bouman and Fehr introduced in \cite{bouman2010sampling} for their proof.  Such relations, though still relatively new, have shown to hold numerous benefits in several applications including higher bit generation rates for random number generation \cite{yao2022quantum} (only shown there for un-biased measurements) along with new applications and easier proofs for high-dimensional systems \cite{krawec2020new}.  They  have been shown to be useful in proving security of quantum cryptographic protocols where standard relations such as Equation \ref{eq:std-eu} actually fail (i.e., prior relations show a trivial bound of $0$ whereas sampling based entropic uncertainty methods show a positive bound) \cite{bae2021source,bae2022quantum}.

In this work, we use the sampling-based approach to derive a novel entropic uncertainty relation for cases where user measurements are biased.  This can occur due to faulty measurement devices for example or, perhaps, ``cheaper'' measurement devices are used which cannot perform an exact measurement in a mutually unbiased basis.  It is also interesting from a theoretical point of view as we prove, here, that better bounds on min entropy are possible even if the two measurement bases are ``close'' to one another.  Finally, it shows even more advantages to the sampling-based approach to entropic uncertainty and we suspect our proof methods here may be highly beneficial to other scenarios where measurement or source devices are imperfect.

We note that, while the main contribution of this paper is our new entropic uncertainty bound, we also make other contributions along the way.  We prove an interesting result (Lemma \ref{lemma:Bell-ent}), that may be independently useful,  which bounds the min entropy of a particular superposition state.  We also prove that higher bit generation rates are possible for BB84 with faulty source and measurement devices and higher bit generation rates are possible for a particular quantum random number generation (QRNG) protocol.  Finally, our main results can be easily incorporated into other quantum cryptographic protocols.

\section{Preliminaries}

We begin by introducing some notation that we use throughout this paper.  We denote by $\al_d$ to be a $d$-character alphabet; without loss of generality we simply assume $\al_d = \{0, 1, \cdots, d-1\}$.  Given a word $q \in \al_d^n$, and some subset $t \subset \{1, 2, \cdots, n\}$, we write $q_t$ to mean the substring of $q$ indexed by $t$, that is $q_t = q_{t_1}q_{t_2}\cdots q_{t_{|t|}}$.  We write $q_{-t}$ to mean the substring of $q$ indexed by the complement of $t$.  Finally, for $i = 1, 2, \cdots, n$, we write $q_i$ to mean the $i$'th character of $q$.

Let $a, b \in \al_d^n$.  We write $\num_i(a)$ to be the number of times the character $i$ appears in $a$.  Formally $\num_i(a) = |\{\ell \st a_\ell = i\}|$.  We extend this to multiple counts in the obvious way, for example $\num_{i,j}(a)$ is the number of times the character $i$ and $j$ appear in $a$, or $\num_{i,j}(a) = |\{\ell \st a_\ell = i \text{ or } a_\ell = j\}|$.  For a bit string $x \in \{0,1\}$, we denote by $w(x)$ to be the relative Hamming weight, namely $w(x) = \num_1(x)/|x|$.  Finally, we denote by $\hd(a,b)$ to be the Hamming distance of words $a$ and $b$, namely: $\hd(a,b) = |\{\ell \st a_\ell \ne b_\ell\}|$.

Given a random variable $X$, we denote by $H(X)$ to be the Shannon entropy of $X$.  If $X$ takes outcome $x_i$ with probability $p_i$, then $H(X) = -\sum_ip_i\log_2p_i$.  Note that all logarithms in this paper are base two unless otherwise specified.  If $X$ is a two outcome random variable taking $x_1$ with probability $p$, then we use $h(p)$ to denote the binary entropy and $H(X) = h(p) = -p\log p - (1-p)\log (1-p)$.  We also define the bounded binary entropy function $\hat{h}(x)$, where $\hat{h}(x) = h(x)$ whenever $x < 1/2$ and $\hat{h}(x) = 1$ otherwise.

A density operator $\rho$ is a Hermitian positive semi-definite operator of unit trace acting on some Hilbert space $\mathcal{H}$.  If $\rho_{AE}$ acts on Hilbert space $\mathcal{H}_A\otimes\mathcal{H}_E$, we write $\rho_A$ to mean the state resulting from tracing out the $E$ system, namely $\rho_A = tr_E\rho_{AE}$.  This is similar for multiple systems.  Given a pure state $\ket{\psi}$ we write $\kb{\psi}$ to mean $\kb{\psi} = \ket{\psi}\bra{\psi}$.  We also define $P(\ket{z})$ to be $P(\ket{z}) = \kb{z}$.  Given an orthonormal basis $\mathcal{B} = \{\ket{v_0}, \cdots, \ket{v_{d-1}}\}$, we write $\ket{i}^\mathcal{B}$ to mean $\ket{v_i}$.  Given $i \in \al_d^n$, we write $\ket{i}^\mathcal{B}$ to mean $\ket{v_{i_1},\cdots,v_{i_n}}$, namely the word $i$ represented in the $\mathcal{B}$ basis.  If the basis is not specified, then it is assumed to be the computational basis $\{\ket{0}, \cdots, \ket{d-1}\}$.  Finally, we use $\ket{\phi_i}$ to denote the Bell states:
\begin{align*}
\ket{\phi_0} = \frac{1}{\sqrt{2}}(\ket{00} + \ket{11}) && \ket{\phi_1} = \frac{1}{\sqrt{2}}(\ket{00} - \ket{11})\\
\ket{\phi_2} = \frac{1}{\sqrt{2}}(\ket{01} + \ket{10}) && \ket{\phi_3} = \frac{1}{\sqrt{2}}(\ket{01} - \ket{10})\\
\end{align*}

Given $\rho_{A}$ we write $H(A)_\rho$ to mean the von Neumann entropy of $\rho_A$, namely $H(A)_\rho = -tr(\rho_A\log\rho_A)$.  Given $\rho_{AE}$, we write $H(A|E)_\rho$ to be the conditional von Neumann entropy, namely $H(A|E)_\rho = H(AE)_\rho - H(E)_\rho$.  We write $\Hmin(A|E)_\rho$ to be the \emph{conditional quantum min entropy} defined to be \cite{renner2008security}:
\begin{equation}
  \Hmin(A|E)_\rho = \sup_{\sigma_E}\max\left\{\lambda\in\mathbb{R} \st 2^{-\lambda}I_A\otimes\sigma_E - \rho_{AE} \ge 0\right\},
\end{equation}
where $A\ge 0$ is used to denote that $A$ is positive semi-definite.  The \emph{smooth conditional min entropy} is denoted $\Hmin^\epsilon(A|E)_\rho$ and is defined to be:
$\Hmin^\epsilon(A|E)_\rho = \sup_{\sigma_{AE}} \Hmin(A|E)_\sigma,$
where the supremum is over all density operators $\sigma_{AE}$ such that $\trd{\sigma_{AE} - \rho_{AE}} \le \epsilon$.  Here we use $\trd{A}$ to mean the \emph{trace distance} of operator $A$.

Quantum min entropy is a very important quantity to measure in quantum cryptography as it relates directly to how many uniform random secret bits may be extracted from a quantum state \cite{renner2008security}.  In detail, assume $\rho_{AE}$ is a \emph{classical-quantum state} (or cq-state).  That is, the $A$ register is classical while the $E$ portion is potentially quantum, thus $\rho_{AE} = \sum_a p(a) \kb{a}\otimes\rho_E^a$.  Assume the $A$ register is $N$-bits in size (i.e., $a \in \{0,1\}^N$ in the sum).  Then \emph{privacy amplification} is a process of picking a random two-universal hash function $f:\{0,1\}^N\rightarrow\{0,1\}^\ell$ and disclosing the choice to Eve, then hashing the $A$ register to $f(A)$ yielding cq-state $\sigma_{KE'}$.  The state $\sigma_{KE'}$ satisfies the following inequality as proven in \cite{renner2008security}:
\begin{equation}\label{eq:PA}
  \trd{\sigma_{KE'} - \uniform_\ell\otimes \sigma_{E'}} \le 2^{-\frac{1}{2}(\Hmin^\epsilon(A|E)_\rho - \ell)} + 2\epsilon,
\end{equation}
where $\uniform_\ell = I/2^\ell$ is a uniform random string of size $\ell$-bits independent of Eve.  Thus, to determine how large $\ell$ can be, one requires a bound on the quantum min entropy \emph{before} privacy amplification.

\subsection{Properties of Quantum Min Entropy}

Min entropy has several properties that we will utilize later.  In particular, given a cqc-state or qqc-state of the form $\rho_{AEC} = \sum_cp(c)\kb{c}\otimes\rho_{AE}^{(c)}$, then:
\begin{equation}\label{eq:min-ent-cl}
\Hmin(A|E)_\rho \ge \Hmin(A|EC)_\rho \ge \min_c \Hmin(A|E)_{\rho^{(c)}}.
\end{equation}
The above is easily shown using the definition of min entropy.  Informally it says that, conditioning on certain events $C$ happening, the min entropy is the ``worst-case'' min entropy of each individual sub-event.


The following lemma from \cite{bouman2010sampling} lets us bound the min entropy in a superposition as a function of the min entropy of a mixed state, assuming the superposition does not have ``too many'' terms:
\begin{lemma}\label{lemma:superposition}
  (From \cite{bouman2010sampling}, based on a lemma in \cite{renner2008security}): Given two orthonormal bases $Z$ and $X$ of some Hilbert space $\mathcal{H}_A$, let $\ket{\psi}_{AE}$ be some pure state of the form $\ket{\psi}_{AE} = \sum_{i\in J}\alpha_i\ket{i}^Z\otimes\ket{E_i}$ where the $\ket{E_i}$ states are arbitrary, but normalized.  Then, if we define the mixed state $\rho_{AE} = \sum_{i\in J} |\alpha_i|^2\kb{i}^Z\otimes\kb{E_i}$, it holds that:
  \[
  \Hmin(X|E)_\psi \ge \Hmin(X|E)_\rho - \log_2|J|,
  \]
  where the $X$ registers, above, are produced by measuring the $A$ register (originally written in the $Z$ basis above), in the $X$ basis.
\end{lemma}

The next lemma we need is from \cite{krawec2022security} and shows how one may compute the min entropy in a state that is initially close to another (in trace distance) but after conditioning on an outcome (after which, the states may no longer be close and, thus, smooth  min entropy by itself cannot be used):
\begin{lemma}\label{lemma:prob}
  (From \cite{krawec2022security}): Let $\rho, \sigma$, and $\tau$, be three quantum states with $\rho$ and $\sigma$ acting on the same Hilbert space ($\tau$ may be arbitrary or trivial).  Also, let $\mathcal{F}$ be a CPTP map with the property that:
  \begin{align*}
    \mathcal{F}(\tau\otimes\rho) &= \sum_xp(x)\kb{x}\otimes \rho_{AE}^{(x)}\\
    \mathcal{F}(\tau\otimes\sigma) &= \sum_xq(x)\kb{x}\otimes \sigma_{AE}^{(x)}.
  \end{align*}
  Then, if $\frac{1}{2}\trd{\rho-\sigma}\le \epsilon$, it holds that:
  \[
  Pr\left(\Hmin^{4\epsilon+3\epsilon^{1/3}}(A|E)_{\rho^{(x)}} \ge \Hmin(A|E)_{\sigma^{(x)}}\right) \ge 1 - 2\epsilon^{1/3},
  \]
  where the probability is over the random outcome $X$ in the above states.
  
\end{lemma}

Finally, we prove the following lemma below in this work which may be of independent interest.  It bounds the min entropy of a quantum state that is a superposition of Bell states on which we have some, but not all, information on (and, thus, Lemma \ref{lemma:superposition} could not be used directly as that lemma requires full information on the superposition size which our lemma below does not require):
\begin{lemma}\label{lemma:Bell-ent}
  Given $\ket{\psi} = \sum_{i \in J}\alpha_i\ket{\phi_i}\ket{E_i}$, where $J = \left\{i\in\al_4^n\st \frac{1}{n}\num_{1,3}(i) \le Q\right\}$, let $\rho_{AE}$ be the result of measuring the first particle of each Bell pair in the $Z$ basis (resulting in register $A$) and tracing out the second particle of each Bell pair.  Then it holds that:
  \begin{equation}
    \Hmin(A|E)_\rho \ge n\left(1 - \hat{h}\left(Q\right)\right).
  \end{equation}
\end{lemma}
\begin{proof}
  We may rewrite $\ket{\psi}$ by permuting subspaces such that the second particle of each Bell pair is ``pushed'' to the left-most subspace while the first particle of each pair is pushed to the middle register (the right-most register will remain $E$).  Noting that $\ket{\phi_0}$ and $\ket{\phi_2}$ are of the form $\frac{1}{\sqrt{2}}(\ket{+,+} \pm \ket{-,-})$ while $\ket{\phi_1}$ and $\ket{\phi_3}$ are of the form $\frac{1}{\sqrt{2}}(\ket{+,-} \pm \ket{-,+})$, the state, after this permutation of subspaces, can be written in the form:
  \begin{equation}
    \ket{\psi} \cong \sum_{b\in\{0,1\}^n}\beta_b\ket{b}^X\otimes\sum_{\substack{a\in\{0,1\}^n\\\frac{1}{n}\hd(a,b)\le Q}}\beta_{a|b}\ket{a}^X\ket{E_{a|b}}.
  \end{equation}
Above, $X$ is the usual Hadamard basis.  From this, we trace out the left-most register (which was originally the second particle of each Bell pair) - this, of course, is equivalent to first measuring the system and then tracing it out - yielding the state:
  \begin{equation}
    \rho_{RE} = \sum_b|\beta_b|^2\underbrace{P\left(\sum_{\substack{a\in\{0,1\}^n\\\frac{1}{n}\hd(a,b)\le Q}}\beta_{a|b}\ket{a}^X\ket{E_{a|b}}\right)}_{\rho^b_{RE}},
  \end{equation}
  where, recall, $P(\ket{z}) = \kb{z}$.

  The $R$ system is now measured in the $Z$ basis yielding $\sum_{b}|\beta_b|^2\rho_{AE}^b$.  From Equation \ref{eq:min-ent-cl}, we have $\Hmin(A|E)_\rho \ge \min_b\Hmin(A|E)_{\rho^b}$.  From Lemma \ref{lemma:superposition}, we have:
  \[
  \Hmin(A|E)_{\rho^b} \ge n - \log\left|\left\{a\in\{0,1\}^n\st \frac{1}{n}\hd(a,b) \le Q\right\}\right|.
  \]
  Noting that, for any $b$, the size of the set $\left\{a\in\{0,1\}^n\st \frac{1}{n}\hd(a,b) \le Q\right\}$ can be bounded using the well-known bound on the size of a Hamming ball, namely
  \[
  \left|\left\{a\in\{0,1\}^n\st \frac{1}{n}\hd(a,b) \le Q\right\}\right| \le 2^{n\hat{h}(Q)},
  \]
  completes the proof.
\end{proof}

\subsection{Quantum Sampling}\label{sec:sample}
Our new entropic uncertainty relation is a so-called \emph{sampling based entropic uncertainty relation} \cite{yao2022quantum} which relies, for its proof, on the quantum sampling framework introduced by Bouman and Fehr in \cite{bouman2010sampling}.  Since we use this framework to prove our main result, we highlight some of the main concepts here.  For more information, the reader is referred to the original sampling paper \cite{bouman2010sampling} from which all information in this section is derived.

A \emph{classical sampling strategy} over $\al_d^N$ is a triple $(P_T, g, r)$, where $P_T$ is a probability distribution over subsets of $\{1, 2, \cdots, N\}$; $g$ is a ``guess function,'' $g:\al_d^* \rightarrow \mathbb{R}$; and $r$ is a ``target function,'' $r:\al_d^*\rightarrow \mathbb{R}$.  Given a word $q \in \al_d^N$, the strategy will first sample $t$ according to $P_T$, observe $q_t$ and compute $g(q_t)$ (or, equivalently, simply observe $g(q_t)$), and use this as a guess for the value of $r(q_{-t})$.  That is, given an observed portion of $q$, the strategy should use that to guess at the target value of an \emph{unobserved} portion of the string.

Let $\delta > 0$, then we define the set of \emph{ideal words} to be:
\[
\mathcal{G}_t = \{q\in\al_d^N \st g(q_t) \dc r(q_{-t})\},
\]
where we write $x\dc y$ to mean $|x - y| \le \delta$.  Then, the \emph{error probability} of the sampling strategy is defined to be:
\begin{equation}
  \epsilon^{cl} = \max_{q\in\al_d^N}Pr\left(q \not \in \mathcal{G}_t\right),
\end{equation}
where the above probability is over the choice of subset $t$.  It is clear from this definition that, for any $q\in\al_d^N$, the probability that the given sampling strategy fails to give a $\delta$-close guess of the target value is at most $\epsilon^{cl}$.  Note that the ``cl'' superscript is used here as a reminder that this is the classical failure probability.

A sampling strategy as above may be promoted to a quantum one.  Let $B$ be a $d$-dimensional orthonormal basis and let $\ket{\psi}_{AE}$ be some quantum state where the $A$ portion lives in a $d^N$ dimensional Hilbert space.  Note that the state $\ket{\psi}$ may be arbitrary.  Then the sampling strategy will first choose a subset $t$ according to $P_T$, and then measure those systems in $A$ indexed by $t$ using basis $B$ to produce outcome $q_t\in\al_d^{|t|}$.  The unmeasured portion collapses to some state $\ket{\psi^t_q}$.  Bouman and Fehr's main result is to give a rigorous analysis of this post measured state.

Formally, we define a space of \emph{ideal states for subset $t$ with respect to basis $B$} (or simply \emph{ideal states} when the context is clear) as follows:
\[
\text{span}\left(\mathcal{G}_t\right)\otimes \mathcal{H}_E  = \text{span}\{\ket{q}^B \st q \in \mathcal{G}_t\}\otimes\mathcal{H}_E.
\]
Note that the definition depends on the chosen basis $B$.  An ``ideal state for subset $t$'' (with respect to basis $B$) is one that lives in this space.  In general, if a $B$ basis measurement is performed on subset $t$ of an ideal state, yielding outcome $q$, then it is guaranteed that the post-measured state is of the form:
\[
\ket{\psi^t_q} = \sum_{i \in J_q}\alpha_i\ket{i}^B\otimes\ket{E_i},
\]
where $J_q = \{i\in\al_d^{N-|t|} \st g(q) \dc r(i)\}$.  Bouman and Fehr's main result is stated in the Theorem below:
\begin{theorem}\label{thm:sample}
  (From \cite{bouman2010sampling}, though we reword it here for our application): Given a classical sampling strategy with error probability $\epsilon^{cl}$ for a given $\delta > 0$, it holds that for any  $\ket{\psi} \in \mathcal{H}_A\otimes\mathcal{H}_E$ (where $\mathcal{H}_A$ is a $d^N$ dimensional Hilbert space) and any $d$-dimensional orthonormal basis $B$, that there exists a collection of ideal states $\{\ket{\phi^t}\}$, indexed by every possible subset choice $t$, such that $\ket{\phi^t}$ are ideal states for subset $t$ with respect to basis $B$, and it holds that:
  \begin{equation}
    \frac{1}{2}\trd{\sum_tP_T(t)\kb{t}\otimes\left(\kb{\psi} - \kb{\phi^t}\right)} \le \sqrt{\epsilon^{cl}}.
  \end{equation}
\end{theorem}

The proof of the above theorem is actually by construction where the ideal states are defined by projecting onto the ideal subspace and a subspace orthogonal to it.  In particular, given a fixed $t$ and an input state $\ket{\psi} = \sum_i\ket{i}^B\otimes\ket{E_i}$, then the ideal states are defined by:
\begin{align*}
  \ket{\psi} &= \braket{\phi^t|\psi}\ket{\phi^t} + \braket{\bar{\phi}^t|\psi}\ket{\bar{\phi}^t}\\
  &= \alpha\sum_{i \in \mathcal{G}_t}\ket{i}^B\otimes\ket{E_i} + \beta\sum_{i\not\in\mathcal{G}_t}\ket{i}^B\otimes\ket{E_i}.
\end{align*}
Thus, given some property of Eve's ancilla in the real state, those properties may translate also to the ideal system, a point that will be important in the proof of our main theorem.

We comment on a few things.  First, Theorem \ref{thm:sample} let's us promote classical sampling strategies to quantum ones where the error (in terms, now, of trace distance) only increases quadratically.  Second, one doesn't actually have to perform the sampling strategy in the given basis - the above states exist regardless.  Thus, one may use the existence of these states but actually perform different measurements on them, yet still be able to say something about the post-measured state.  We will use this later in our proof.  Finally, though our wording of Theorem \ref{thm:sample} is different from how it was worded originally in \cite{bouman2010sampling}, their original proof is by construction and readily leads to the above statement as shown in \cite{yao2022quantum}.

Before leaving this section, we discuss a basic sampling strategy for bit strings (i.e., $d=2$).  Let $P_T$ be the uniform distribution on subsets of size $m$ (with $m < N/2$) and let $g(x) = r(x) = w(x)$.  From this, it is clear that the set of ideal words is:
\begin{equation}\label{eq:good-1}
\mathcal{G}_t = \left\{q\in\{0,1\}^N \st w(q_t) \dc w(q_{-t})\right\},
\end{equation}
where $n = N-m$.  Thus, this strategy observes the relative number of $1$'s in the given string $q_t$ and uses this as a guess as to the number of $1$'s in the unobserved portion $q_{-t}$.  Then, it was proven in \cite{bouman2010sampling}, that the error probability of this strategy may be bounded by:
\begin{equation}\label{eq:cl-bound}
  \epsilon^{cl}_0 \le 2\exp\left(-\delta^2\frac{mN}{N+2}\right).
\end{equation}
The above equation will be useful later.

\section{New Entropic Uncertainty Relation}

We now prove our main result. Consider the following experiment.  Let $\rho_{ABE}$ be a quantum state where the $A$ and $B$ registers each consist of $N$ qubits. Also consider two bases $Z$ and $\X$, where $\X$ is defined to be spanned by the states $\ket{x_0} = \alpha\ket{0} + \sqrt{1-\alpha^2}\ket{1}$ and $\ket{x_1} = \sqrt{1-\alpha^2}\ket{0} - \alpha\ket{1}$ where $\alpha = \sqrt{\frac{1}{2}+b}$ for some bias parameter $b \in [-.5, .5]$ (our methods can be extended to arbitrary complex amplitudes $\alpha$, however we restrict to real values for this work as the presentation is simpler and, already, this gives an interesting result as shown later in our evaluations).  Note that, when $b=0$, the $\X$ basis is the usual Hadamard basis.  When $b=\pm1/2$, the $\X$ basis is no different from the $Z$ basis.  We assume $b$ is known or can be bounded by the parties running the experiment.

Given $\rho_{ABE}$, Alice and Bob will choose a random subset $t$ of size $m < N/2$ and measure their qubits, indexed by $t$, in basis $\X$.  Let $q\in\{0,1\}^m$ be the result of XOR'ing their measurement results (i.e., $q_i = 0$ if the $i$'th measurement yielded equal outcomes in the $\X$ basis and it is $1$ otherwise).  This causes the remaining $n=N-m$ qubits to collapse to some state $\rho_{ABE}^{(t,q)}$.  Next, the remaining $n$ qubits are measured in the $Z$ basis.  Our main result is stated in  Theorem \ref{thm:main} below, provides a bound on the min entropy in this $Z$ basis measurement as a function of the bias parameter $b$, and the Shannon entropy of the observed parity value $q$.

Our proof assumes the states under investigation have a specific form on the adversary/environment system as defined below in Definition \ref{def:sym}.  This assumption is needed in only one part of our proof, though removing the assumption does seem to greatly complicate the proof.  We suspect this assumption is not actually required, though a full proof remains elusive.  That being said, the assumption below is, in a way, minimal and, in fact, most quantum states investigated in security proofs satisfy it.  Thus, while we have to make this assumption on the given quantum state, it is not very problematic towards applications, including cryptographic ones.  In fact, this assumption may even be enforced if one utilizes mismatched measurements \cite{barnett1993eavesdropping,watanabe2008tomography,matsumoto2008key,krawec2016asymptotic} (see also methods in \cite{krawec2017quantum}).

\begin{define}\label{def:sym}
  Let $\ket{\psi}_{ABE}$ be a quantum state with the $A$ and $B$ portions consisting of $N$ qubits each.  Without loss of generality, we may write $\ket{\psi}_{ABE} = \sum_{i\in \al_4^N}\alpha_i\ket{\phi_i}\ket{E_i}$, where $\ket{\phi_i}$ is the Bell basis defined earlier.  We say $\ket{\psi}_{ABE}$ is \emph{produced by a depolarizing source} if it holds that $\braket{E_i|E_j} = 0$ whenever $i \ne j$.
\end{define}
Note that a depolarizing channel produces a state according to Definition \ref{def:sym}.  A state produced by a depolarizing source also produces, in a way, ``symmetric'' (though potentially still biased based on the measurements' biases) measurement results and so can even be enforced as mentioned earlier (using, potentially, mismatched measurements if measurements are biased \cite{krawec2017quantum}).

To prove our main result, we'll need the following classical sampling strategy: Given a word $q \in \al_4^{n+m}$, choose a subset $t \subset \{1,\cdots, n+m\}$ of size $|t| = m$ uniformly at random.  The guess function is the relative number of $1$'s and $3$'s in the observed portion, namely $f(q_t) = \frac{1}{m}\num_{1,3}(q_t)$.  The target function is the relative number of $1$'s and $3$'s in the unobserved portion, $r(q_{-t}) = \frac{1}{n}\num_{1,3}(q_{-t})$.  This induces the set of ideal words:
\begin{equation}
  \mathcal{G}_t = \left\{i\in\al_4^{n+m} \st \frac{1}{m}\num_{1,3}(i_t) \dc \frac{1}{n}\num_{1,3}(i_{-t})\right\}.
\end{equation}
The classical error probability of this sampling strategy is analyzed in the following Lemma:

\begin{lemma}\label{lemma:sample-used}
  Given $\delta > 0$ and $m < n$, the classical error probability $\epsilon^{cl}$ of the sampling strategy described above is bounded by:
  \[
  \epsilon^{cl} \le 2\exp\left(-\delta^2\frac{m(n+m)}{n+m+2}\right).
  \]
\end{lemma}
\begin{proof}
  We prove this by, essentially, reducing to the sampling strategy described at the end of Section \ref{sec:sample} and bounded by Equation \ref{eq:cl-bound}.  Let $\widetilde{\mathcal{G}}_t$ be the set of ideal words for the earlier defined sampling strategy (see Equation \ref{eq:good-1}).  Let $q \in \al_4^N$ (with $N= n+m$) and consider a fixed subset $t$ of size $m<N/2$.  Then, define the word $\tilde{q} \in \{0,1\}^N$ where $\widetilde{q}_i = 0$ if $q_i = 0$ or $2$ and $\widetilde{q}_i = 1$ otherwise.  Thus, $w(\widetilde{q_t}) = \frac{1}{m}\num_{1,3}(q_t)$ and, similarly, for the complement of $t$.  In particular, for any $t$, it holds that $q \not \in \mathcal{G}_t \iff \widetilde{q}\not\in\widetilde{\mathcal{G}}_t$. From this, we conclude:
  \[
  Pr(q \not \in \mathcal{G}_t) = Pr\left(\widetilde{q} \not \in \widetilde{\mathcal{G}}_t\right) \le \max_{i\in\{0,1\}^N} Pr\left(i \not \in \widetilde{\mathcal{G}}_t\right) \le \epsilon_0^{cl},
  \]
  where $\epsilon_0^{cl}$ was defined in Equation \ref{eq:cl-bound}.  Since the above is true for any $q$, the proof is complete.
\end{proof}

We now have all the tools we need to state and prove our main result:

\begin{theorem}\label{thm:main}
  Let $\epsilon > 0$, $\alpha=\sqrt{1/2 + b}$ for some $b\in[-.5,.5]$, and let $\ket{\psi}_{ABE}$ be a state prepared by a depolarizing source (according to Definition \ref{def:sym}) where the $A$ and $B$ registers each consist of $N$ qubits.  Assume a random subset is chosen $t$ of size $m$ and a measurement in the $\X$ basis is performed in the $A$ and $B$ registers, indexed by $t$ and resulting in outcomes $q_A, q_B\in\{0,1\}^m$.  The remaining qubits in the $A$ and $B$ portions are measured in the $Z$ basis resulting in state $\rho_{ABE}^{(t,q)}$ (which depends on $q=q_A\oplus q_B$ and $t$).  Then, except with probability $\epsilon_{fail} = (16\epsilon)^{1/3}$, it holds that:
  \begin{equation}
    Pr\left(\Hmin^{8\epsilon+3(2\epsilon)^{1/3}}(A|E)_{\rho^{(t,q)}} \ge n(1 - \hat{h}(w(q) + \nu + \delta))\right),
  \end{equation}
  where $\hat{h}(x)$ is the bounded binary entropy function and:
  \begin{equation}\label{eq:thm:delta}
    \delta = \sqrt{\frac{(m+n+2)}{m(m+n)}\ln\frac{2}{\epsilon^2}}
  \end{equation}
  and
  \begin{equation}
    \nu = 4b^2 + \frac{1}{\sqrt{m}}\ln\frac{1}{2\epsilon}
  \end{equation}
  The probability is over the choice of subset and the measurement outcome $q=q_A\oplus q_B$.
\end{theorem}
\begin{proof}
  Let $\epsilon > 0$ be given and set $\delta$ as in Equation \ref{eq:thm:delta}.  From Theorem \ref{thm:sample}, and using the sampling strategy described earlier in this section and analyzed in Lemma \ref{lemma:sample-used}, there exist ideal states $\{\ket{\phi^t}_{ABE}\}$, with respect to the Bell basis, such that $\ket{\phi^t} \in \text{span}\left(\mathcal{G}_t\right)\otimes\mathcal{H}_E$ where:
  \[
  \text{span}\left(\mathcal{G}_t\right) = \text{span}\left\{\ket{\phi_q} \st \frac{1}{m}\num_{1,3}(q_t) \dc \frac{1}{n}\num_{1,3}(q_{-t})\right\}
  \]
  and:
  \[
  \trd{\sum_tP_T(t)\kb{t}\otimes(\kb{\psi} - \kb{\phi^t})} \le \sqrt{\epsilon^{cl}} \le \epsilon,
  \]
  where the latter inequality follows from Lemma \ref{lemma:sample-used} and our choice of $\delta$.
  
  Note that, since these states are constructed by projecting $\ket{\psi}$ into the subspace of ideal states, it is not difficult to see that, since $\ket{\psi}$ is produced by a depolarizing source, each $\ket{\phi^t}$ is also.  (See the discussion under Theorem \ref{thm:sample}.)

  We first analyze the ideal states and show the min entropy there is high, based on the observed $\X$ basis noise.
  
  By permuting subspaces so that those systems indexed by $t$ are the left-most system, we may write:
  \begin{equation}\label{eq:ideal-pre-m}
  \ket{\phi^t} \cong \sum_{i\in\al_4^m}\alpha_i\ket{\phi_i} \otimes \underbrace{\sum_{\ell \in J_i} \beta_{\ell|i}\ket{\phi_\ell}\ket{E_{i,j}}}_{\ket{\mu_i}},
  \end{equation}
  with:
  \[
  J_i = \left\{\ell\in\al_4^n \st \frac{1}{n}\num_{1,3}(\ell) \dc \frac{1}{m}\num_{1,3}(i)\right\}.
  \]
  Note that we are permuting subspaces only for clarity in presentation, this is not a required step of the protocol.    Now, if we were able to make a Bell basis measurement on subset $t$, observing, say, outcome $x \in \al_d^m$, we would know, for certain, that the post measured state must have collapsed to $\ket{\phi^t_x} = \sum_{y}\beta_y\ket{\phi_x}\ket{E_x}$ where the number of $1$'s and $3$'s in $y$ is $\delta$-close to the number of $1$'s and $3$'s in the observed $x$.  However, we can only measure in the $\X$ basis leading to outcomes $q_A$ and $q_B$.  The idea is that, based on $\alpha$, the observed string cannot be too different from the underlying state in the original Bell basis.  To prove this formally, we now consider the following two-qubit basis based on $\X$ (which we call the $\X$-Bell basis):
  \begin{align*}
  \ket{\phi^X_0} &= \frac{1}{\sqrt{2}}\ket{x_0,x_0} + \frac{1}{\sqrt{2}} \ket{x_1,x_1}\\
  \ket{\phi^X_1} &= \frac{1}{\sqrt{2}}\ket{x_0,x_1} + \frac{1}{\sqrt{2}} \ket{x_1,x_0}\\
  \ket{\phi^X_2} &= \frac{1}{\sqrt{2}}\ket{x_0,x_0} - \frac{1}{\sqrt{2}} \ket{x_1,x_1}\\
  \ket{\phi^X_3} &= \frac{1}{\sqrt{2}}\ket{x_0,x_1} - \frac{1}{\sqrt{2}} \ket{x_1,x_0}\\
  \end{align*}
  Note that if $\alpha = 1/\sqrt{2}$ (thus $\X$ basis is the Hadamard basis), then it holds $\ket{\phi^X_i} = \ket{\phi_i}$ for $i=0,1,2,3$.
  
  Changing basis of those systems indexed by $t$ in Equation \ref{eq:ideal-pre-m}, we have:
  \begin{align}
  \ket{\phi^t} &\cong \sum_{i\in\al_4^m}\alpha_i\left(\sum_{j\in\al_4^m}\braket{\phi_j^X|\phi_i}\ket{\phi_j^X}\right)\otimes \ket{\mu_i}\notag\\
  &= \sum_{j\in\al_4^m} \ket{\phi_j^X}\otimes\left(\sum_{i\in\al_4^m}\alpha_i\braket{\phi_j^X|\phi_i}\ket{\mu_i}\right).\label{eq:ideal-pre-m-2}
  \end{align}
  
  A measurement is now performed on the $A$ and $B$ registers, indexed by $t$, in the $\X$ basis.  However, the important factor will be the number of errors in the measurements.    Thus, we equivalently consider Alice and Bob measuring in the following two-outcome POVM: $X_0 = \kb{x_0,x_0} + \kb{x_1,x_1}$ and $X_1 = \kb{x_0,x_1} + \kb{x_1,x_0}$.  Thus, $X_1$ represents an outcome where Alice and Bob get different measurement outcomes after measuring in basis $\X$.  Note that an outcome of $X_1$ can only occur if the underlying state is $\ket{\phi^X_1}$ or $\ket{\phi^X_3}$.  Of course, if $\alpha=1/\sqrt{2}$ and $\X$ is the Hadamard basis, this gives us an exact count of the number of $1$'s and $3$'s in the state $i$ (needed to bound the entropy in $\ket{\mu_i}$).  However, we actually only count the number of $1$'s and $3$'s in $j$ - from this, we will need to determine a good bound for the number of 1's and 3's in $i$.  Intuitively, this should follow since, for $\alpha$ close to $1/\sqrt{2}$, the $\X$-Bell states are almost the Bell states and, so, any entropy equation should behave similarly in both bases for small bias parameter $b$.  We prove this rigorously below.
  
  For a fixed $j\in\al_4^m$, and user-defined $\nu \ge 0$, let's define ``good''  and ``bad'' states as follows:
  \begin{align*}
  G_j &= \{i \in \al_4^m \st \hd(i,j) \le m\nu\}\\
  B_j &= \{i \in \al_4^m \st \hd(i,j) > m\nu\}.
  \end{align*}
  Note that $\nu$ will control how likely we are to get a ``good'' state as larger $\nu$ means more states are considered good - though this will lead to additional uncertainty in $i$ as we also want to control how far $i$ is from $j$. We will show later that $\nu$ may be  made a function of $\epsilon$.   Given this, we may rewrite Equation \ref{eq:ideal-pre-m-2} as follows: $\ket{\phi^t} \cong $
  \begin{equation}\label{eq:phi-cob}
  \sum_{j\in\al_4^m} \ket{\phi_j^X}\otimes\left( \sum_{i \in G_j}\alpha_i\braket{\phi^X_j|\phi_i}\ket{\mu_i} + \sum_{i \in B_j}\alpha_i\braket{\phi^X_j|\phi_i}\ket{\mu_i}\right)
  \end{equation}
  
  Let $\ket{g_j} = \sum_{i \in G_j}\alpha_i\braket{\phi^X_j|\phi_i}\ket{\mu_i}$ and $\ket{b_j} = \sum_{i \in B_j}\alpha_i\braket{\phi^X_j|\phi_i}\ket{\mu_i}$ and so $\ket{\phi^t} \cong \sum_j\ket{\phi_j^X}\otimes\left(\ket{g_j} + \ket{b_j}\right)$.
  
  We now consider an ``ideal-ideal'' state $\ket{\widetilde{\phi}^t}$ defined as:
  \begin{equation}\label{eq:ideal-ideal}
  \ket{\widetilde{\phi}^t} = \frac{1}{\sqrt{M}}\sum_{j\in\al_4^m}\ket{\phi^X_j}\otimes\ket{g_j},
  \end{equation}
  where $M = \sum_j\bk{g_j}$.  By basic properties of trace distance, we have:
  \begin{equation}
  \frac{1}{2}\trd{\kb{\phi^t} - \kb{\widetilde{\phi}^t}} = \sqrt{1 - |\braket{\phi^t|\widetilde{\phi}^t}|^2}.
  \end{equation}
  Since all states are prepared by a depolarizing source, we have:
  \begin{align*}
  1 - |\braket{\phi^t|\widetilde{\phi}^t}|^2 &= 1 - \left|\frac{1}{\sqrt{M}}\sum_j\left(\braket{g_j|g_j} + \braket{g_j|b_j}\right)\right|^2\\
  &= 1 - \frac{1}{M}\left(\sum_j\bk{g_j}\right)^2 = 1-M.
  \end{align*}

  We claim that $1-M$ may be bounded above by an arbitrarily small value if user parameters are set appropriately.  Note that $1-M = \sum_j\bk{b_j}$.  This follows from the fact that Equation \ref{eq:phi-cob} is normalized and so:
  \[
  1 = \sum_j(\bk{g_j} + \bk{b_j}) \Longrightarrow \sum_j\bk{b_j} = 1-M.
  \]
  Now, since the state is produced by a depolarizing source, we find:
  \begin{align}
    \sum_{j\in\al_4^m}\bk{b_j} &= \sum_{j\in\al_4^m}\left(\sum_{i\in B_j}|\alpha_i|^2|\braket{\phi_j^X|\phi_i}|^2\right)\notag\\
    &=\sum_{i\in\al_4^m}|\alpha_i|^2\sum_{j \in B_i}|\braket{\phi_j^X|\phi_i}|^2.\label{eq:sum-bj}
  \end{align}

  For a fixed $i \in \al_4^m$, let's focus on $\sum_{j\in B_i}|\braket{\phi_j^X|\phi_i}|^2$.  The following identities can be easily shown for any $\alpha \in [0,1]$:
  \begin{align*}
    \ket{\phi_0^X} &= \ket{\phi_0}, &&    \ket{\phi_3^X} = \ket{\phi_3},\\
    \ket{\phi_1^X} &= \sqrt{p}\ket{\phi_2} + \sqrt{q}\ket{\phi_1},\\
     \ket{\phi_2^X} &= \sqrt{q}\ket{\phi_2} - \sqrt{p}\ket{\phi_1}
  \end{align*}
  where:
  \begin{align*}
    \sqrt{p} = \beta^2 - \alpha^2 = 2b, && \sqrt{q} = 2\alpha\beta
  \end{align*}
  From this, we see that, given a fixed $i\in\al_4^m$, and a particular $j\in B_i$, then if there exists even a single index $\ell \in \{1, 2, \cdots, m\}$ such that $i_\ell = 0$ and $j_\ell \ne 0$ or $i_\ell = 3$ and $j_\ell \ne 3$, then the entire inner product $\braket{\phi_j^X|\phi_i} = 0$.  Since we want to upper-bound Equation \ref{eq:sum-bj}, the only way that expression can have non-zero terms is if, for a given $i$, $j_\ell=0$ whenever $i_\ell=0$ and $j_\ell=3$ whenever $i_\ell = 3$.  If $i_\ell = 1$ or $2$, then $j_\ell$ may be either $1$ or $2$.  Of course, since we are summing over ``bad'' states, there must be at least $m\nu$ differences in $j$.

  Considering any fixed $i$, if $\num_{1,2}(i) \le m\nu$, it is clear that $\sum_{j\in B_i}|\braket{\phi_j^X|\phi_i}|^2 = 0$ since at least one index in each $j\in B_i$ must differ on an index where $i_\ell = 0$ or $3$.  The only time the sum over $j$ can be non-zero is if $i$ satisfies $\num_{1,2}(i) = k > m\nu$.  For any such $i$, there exists a $j \in B_i$ such that $\hd(i,j) = d$ with $m\nu < d \le k$ and where $j=i$ everywhere except on $d$ indices where $i$ happens to be $1$ ($j$ will be a $2$ on such an index) or $2$ ($j$ will be a $1$ on such an index).  This would lead to a value of $|\braket{\phi_j^X|\phi_i}|^2 = p^dq^{k-d} = p^d(1-p)^{k-d}$, where we note that $q = 1-p$.  The $p^d$ term comes from changing the $d$ indices (flipping a $1$ to a $2$ and a $2$ to a $1$) while the $q^{k-d}$ term comes from leaving the remaining $1$'s and $2$'s in $i$ the same in $j$.  Of course the rest of $i$ are $0$'s and $3$'s which are kept the same in $j$.  Since there are ${k \choose d}$ such strings $j$, it follows that for any $i$ with $\num_{1,2}(i) = k > m\nu$, that $\sum_{j\in B_i}|\braket{\phi_j^X|\phi_i}|^2 = \sum_{d=m\nu}^k{k \choose d}p^d(1-p)^{k-d}$.

Continuing this logic, we can write Equation \ref{eq:sum-bj} in the following way:
\begin{align}
  1-M &= \sum_{i\in\al_4^m}|\alpha_i|^2\sum_{j\in B_i}|\braket{\phi_j^X|\phi_i}|^2\notag\\
  &\le \sum_{k=m\nu}^m\widetilde{p}(k)\sum_{d=m\nu}^k{k\choose d}p^d(1-p)^{k-d},
\end{align}
where:
\[
\widetilde{p}(k) = \sum_{i \st \num_{1,2}(i)=k}|\alpha_i|^2.
\]
Note that, if $m\nu$ is not an integer, we take the floor value and thus the reason for the inequality above.  Note also that $\sum_{k=0}^m\widetilde{p}(k) = \sum_{i\in\al_4^m}|\alpha_i|^2 = 1$.  Thus:
\begin{align}
1-M &\le \max_{k\le m} \left(\sum_{d=m\nu}^k{k\choose d}p^d(1-p)^{k-d}\right)\notag\\
&\le \sum_{d=m\nu}^m{m\choose d}p^d(1-p)^{m-d}.
\end{align}
This can be considered the tail of the CDF of a binomial distribution with parameter $p$ and $m$ trials.  By Hoeffding's inequality, we may derive the following bound, for $\nu \ge p$:
\begin{equation}
  1-M \le \exp\left(-2m(\nu-p)^2\right).
\end{equation}
By setting $\nu = p + \frac{1}{\sqrt{m}}\ln\frac{1}{2\epsilon}$, it holds that $\sqrt{1-M} \le \epsilon$ and thus we have:
\[
\frac{1}{2}\trd{\kb{\phi^t} - \kb{\widetilde{\phi}^t}} \le \sqrt{1-M} \le \epsilon.
\]
Of course, this is true for any subset $t$ in the ideal system $\ket{\phi^t}$ and, so, by the triangle inequality, along with elementary properties of trace distance, we have:
\begin{equation}\label{eq:ideal-ideal-dist}
  \frac{1}{2}\trd{\sum_tP_T(t)\kb{t}\otimes\left(\kb{\psi} - \kb{\widetilde{\phi}^t}\right)} \le 2\epsilon.
\end{equation}
Thus, since the given input state $\ket{\psi}$ is actually $2\epsilon$ close to these ``ideal-ideal'' states, we may analyze the entropy there and use Lemma \ref{lemma:prob} to promote the analysis to the real state.

Define $\sigma_{TQ} = \sum_tP_T(t)\kb{t}\otimes\kb{\widetilde{\phi}^t}$ and we analyze the min entropy in this state, following the conclusion of the measurements and sampling.  Sampling on such a state implies measuring the subset register $T$ causing the state to collapse to $\ket{\widetilde{\phi}^t}$.  After measuring those systems indexed by $t$ in the POVM $X_0$ and $X_1$ defined above, observing $q\in \{0,1\}^m$, then tracing out the measured portion, the state collapses to $\ket{\widetilde{\phi}^t_q}$ which may be written in the form:
\begin{equation}
  \ket{\widetilde{\phi}^t} = \sum_{\substack{j\in\al_4^m \\ \num_{1,3}(j) = \num_1(q)}}p_jP\left(\sum_{i\in G_j}\beta_{i|j}\ket{\mu_i}\right).
\end{equation}
The above can be seen easily from Equation \ref{eq:ideal-ideal} and simply re-parameterizing.  Note that whenever an observation of $X_1$ is observed, the underlying index of $j$ may be either a $1$ or a $3$ and, thus, the state collapses to some $j$ where we have a bound on the number of $1$'s and $3$'s based on the observed $q$.  Let $Q = \num_1(q)$.  Continuing our derivation, we may write the above state in the following form:
\begin{align}
  \ket{\widetilde{\phi}^t} &= \sum_{\substack{j\in\al_4^m \\ \num_{1,3}(j) = Q}}p_jP\left(\sum_{i\in G_j}\beta_{i|j}\ket{\mu_i}\right)\notag\\\notag\\
  &=\sum_{\substack{j\in\al_4^m \\ \num_{1,3}(j) = Q}}p_jP\left(\sum_{\substack{i\in \al_4^m\\\hd(i,j) \le m\nu}} \beta_{i|j}\right.\notag\\
  &\times\left.\left[\sum_{\substack{\ell\in\al_4^n\\\frac{1}{n}\num_{1,3}(\ell) \dc \frac{1}{m}\num_{1,3}(i)}}\gamma_{\ell|i,j}\ket{\phi_\ell}\ket{E_{\ell|i,j}}\right]\right)\notag\\\notag\\
  &= \sum_{\substack{j\in\al_4^m \\ \num_{1,3}(j) = Q}}p_jP\left(\sum_{\substack{\ell\in\al_4^m\\\frac{1}{n}\num_{1,3}(\ell) \le w(q) + \nu + \delta}}\widetilde{\gamma}_{\ell|j}\ket{\phi_{\ell}}\ket{\widetilde{E}_{\ell|j}}\right).
\end{align}
Above, for the last equality, we simply re-parameterized and changed the order of the summation.  Note that some of the $\widetilde{\gamma}_{\ell|j}$ values may be zero.  We did this so that we can easily use Equation \ref{eq:min-ent-cl} along with Lemma \ref{lemma:Bell-ent} to find the following lower-bound:
$  \Hmin(A|E)_{\widetilde{\phi}^t_q} \ge 1 - \hat{h}(w(q) + \nu + \delta),$
where the $A$ register is used to store a $Z$ basis measurement of the first particle of each Bell pair in the above state (the second particle is traced out).

Of course, this is only the ideal state.  However, Equation \ref{eq:ideal-ideal-dist}, along with Lemma \ref{lemma:prob}, finishes the proof.  In particular, the $X$ random variable for Lemma \ref{lemma:prob} is the subset choice $t$ and measurement outcome $q$ while the CPTP map $\mathcal{F}$ is the choice of subset and the measurement in POVM $\{X_0, X_1\}$.

\end{proof}

\begin{corollary}
Let $\rho_{ABE}$ be a quantum state where the $A$ and $B$ registers hold a single qubit.  Let $\alpha = \sqrt{1/2 + b}$ for some $b \in [-.5,.5]$ and let $Q_X$ be the random variable induced by performing an $\X$ basis measurement on the $A$ and $B$ qubit and XOR'ing the outcome.  Let $Q_X^b$ be the random variable which takes the value $1$ with probability $\min(1/2, Pr(Q_X=1) + 4b^2)$.  Then it follows that:
\begin{equation}\label{eq:corr:result}
H(A_Z|E)_\rho + H\left(Q_X^b\right) \ge 1.
\end{equation}
where $A_Z$ is the random variable induced by Alice's $Z$ basis measurement on her particle in $\rho_{ABE}$.
\end{corollary}
\begin{proof}
This follows immediately from Theorem \ref{thm:main} and by the asymptotic equipartition property \cite{tomamichel2009fully} and the law of large numbers.
\end{proof}

\subsection{Comparison to Standard Entropic Uncertainty in the Asymptotic Limit}

In the next section, we apply our new entropic uncertainty bound to two particular cryptographic applications, each of which were proven in previous work, using standard entropic uncertainty relations for quantum min entropy and we compare the resulting bit generation rates for various bias parameters and noise levels in the channel.  However, before this, we show here a comparison in the asymptotic case to the following standard entropic uncertainty inequality proven in \cite{berta2010uncertainty} (written in a form, here, for the particular scenario and measurements we're interested in):
\begin{equation}\label{eq:asym-other}
H(A_Z|E) + H(A_X|B_X) \ge -\log_2\left(\frac{1}{2}+b\right),
\end{equation}
for $b \ge 0$.  Such a comparison gives a general notion of the improvement that is possible using our new result, since the asymptotic case will always provide an upper-bound.

For this comparison, we assume the state is produced by a depolarization channel (which is easily confirmed to satisfy Definition \ref{def:sym}), and thus have $H(A_X|B_X) = h(q)$, where $q$ will denote the error rate in the channel.  Comparing with Equation \ref{eq:corr:result}, of course when $b = 0$, the two identities agree exactly, as expected.

The comparison for $b \ge 0$ is shown in Figure \ref{fig:comp}.  There are several interesting observations to make here; in particular, we note that, in many settings, our entropic uncertainty relation produces a strictly better bound on the entropy.  However, this is not always the case.  In particular, when the noise and bias are both small, standard entropic uncertainty produces a better result.  However, in all other tests we performed when the noise is larger and there is bias, our result produces a strictly better bound on the entropy.  Since both our new result and standard results are both lower-bounds, one may, in practice, simply take the maximum of the two and, thus, our work can only benefit future analyses requiring bounds on quantum entropy with biased measurements.

\begin{figure}
    \centering
    \includegraphics[width=.48\linewidth]{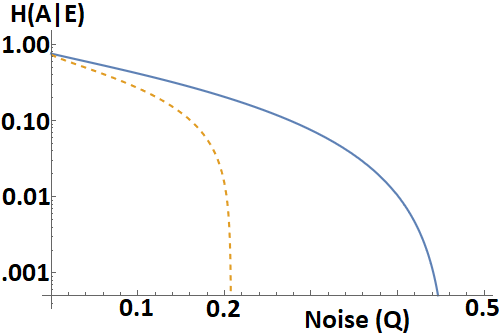}
    \includegraphics[width=.48\linewidth]{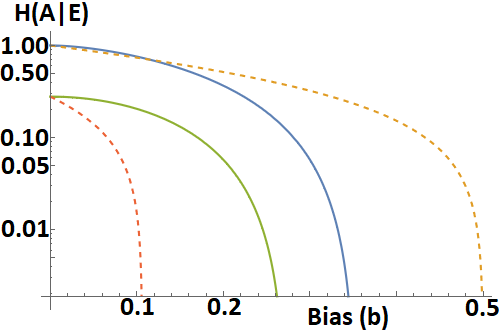}
    \caption{Comparing our new entropic uncertainty relation (Solid lines, Equation \ref{eq:corr:result}) to standard entropic uncertainty relations in the asymptotic limit (Dashed lines, Equation \ref{eq:asym-other}).  Since these are lower-bounds, higher is better here.  Left: Here we fix the bias at $b=.1$ and vary the noise parameter $q$ ($x$-axis) from $0$ to $50\%$.  Right: Here, we fix the noise at $0\%$ (Blue and Yellow) and $20\%$ (Green and Red) as the bias ($x$-axis) varies from $b=0$ to $b=0.5$. Note that our new result produces the same or better results in most settings.  However, when there is no noise, our result tends to perform worse, except for a certain range of bias $b < .2$ as shown in the Right figure (Blue and Yellow comparison).  See text for additional discussion.}
    \label{fig:comp}
\end{figure}


\section{Applications}

We now apply our main theorem to two different cryptographic applications.  The first is a quantum random number generator (QRNG) with a faulty and uncharacterized source.  The second is a QKD protocol where Alice and Bob are not able to measure in mutually unbiased bases, as is typically required by BB84 style protocols to maximize key generation rates.  In both instances we show there are several cases where our new result significantly outperforms prior work using standard entropic uncertainty relations.

\textbf{Quantum Random Number Generation: } We first consider a \emph{source independent} (SI) QRNG protocol whereby the measurement devices are fully characterized, but the source is unknown, as introduced in \cite{vallone2014quantum}.  The goal of a QRNG protocol is to distill a cryptographically secure random bit string from a quantum source.  SI security models offer a nice ``middle ground'' between fully trusted devices (which have weak security guarantees) and fully device independent models, which offer strong security guarantees \cite{colbeck2011private,pironio2013security} but have low bit generation rates with today's technology \cite{bierhorst2018experimentally,liu2018high}. SI-QRNG protocols have been demonstrated experimentally to have high bit generation rates reaching in the Gbps range \cite{avesani1801secure,drahi2020certified}.  For a general survey of QRNG protocols, the reader is referred to \cite{herrero2017quantum}.

Typically SI-QRNG protocols operate by having the uncharacterized source prepare quantum signals and sending them to a user.  The user measures some of the signals in one basis to determine the fidelity of the signal.  The remaining signals are measured in an alternative basis leading to a \emph{raw random string}.  The raw random string may not be truly uniform random and so needs to be further processed through privacy amplification.  If one can bound the quantum min entropy of the raw random string, Equation \ref{eq:PA} may be used to determine the number of bits that may be extracted from the source, even if the source happens to be adversarial.

We analyze the SI-QRNG protocol introduced in \cite{xu2016experimental}.  In this protocol, the source should prepare $N$ copies of the Bell state $\ket{\phi_0}$ and send both particles to Alice.  Alice chooses a random subset and measures both particles in the $\X$ basis (denoting by $q$ as the outcome of the parity of these measurements; namely $q_i=0$ if on the $i$'th test, Alice observed the same outcome, either $\ket{x_0}$ or $\ket{x_1}$, in both particles). For the remaining Bell pairs, Alice measures the first particle in the $Z$ basis, discarding the second particle.  Let $\alpha = \sqrt{\frac{1}{2}+b}$ with $b \ge 0$ (the case when $b < 0$ turns out to be symmetric with the equations we use).  Using a standard entropic uncertainty relation from \cite{tomamichel2011uncertainty}, the authors of \cite{xu2016experimental} were able to derive the following bound on the bit generation rate:
\begin{equation}\label{eq:QRNG-other}
r_{other} = \frac{1}{N}\left(-n\log\left(\frac{1}{2}+b\right) - n\log_2\gamma(w(q)+\delta')\right),
\end{equation}
where:
\[
\gamma(x) = \left(x + \sqrt{1+x^2}\right)\left(\frac{x}{\sqrt{1+x^2}-1}\right)^x,
\]
and
\[
\delta' = 2\sqrt{\frac{N^2}{n^2m}\ln\frac{4}{\epsilon'}}.
\]
Of course, the original work in \cite{xu2016experimental} only considered the case when $b=0$, however since their proof relies on the standard entropic uncertainty relation from \cite{tomamichel2011uncertainty}, it is  not difficult to see it can be applied to any $b$.

Using our Theorem \ref{thm:main}, along with Equation \ref{eq:PA}, we can, instead, derive the following bit generation rate:
\begin{equation}\label{eq:QRNG-ours}
r_{ours} = \frac{1}{N}\left(n(1 - \hat{h}(w(q) + \nu + \delta)) + 2\log\frac{1}{\epsilon}\right).
\end{equation}

In our evaluations, we set a sampling size of $7\%$ (thus $m = 0.07N$) and we set $\epsilon' = 10^{-12}$ (for $r_{other}$) and $\epsilon = 10^{-36}$ (for $r_{ours}$).  This implies a failure probability and a security level on the order of $10^{-12}$ for both equations to make a fair comparison.  Note that in our bound, we require a much smaller $\epsilon$ to guarantee the same level of security as other work - this is a disadvantage to our approach caused by the use of Lemma \ref{lemma:prob}.  However, we will see that even with this disadvantage, our result still produces higher rates in many scenarios.

Figures \ref{fig:QRNG:1} and \ref{fig:QRNG:2} compare the bit generation rates of this protocol using our new result (solid lines) and prior work using standard entropic uncertainty (dashed lines).  We note several things.  First, our new bound produces higher bit generation rates in many of the tested scenarios.  There are times, however, when prior work surpasses ours - in particular when the noise is low, however this was also observed in the previous section.  We conjecture that our methods may be improved in the low noise case, however we leave that as interesting future work.  Regardless, our work provides substantially improved results in many cases and, since these are all lower bounds, users of these protocols with biased measurements may simply take the max of both our work and prior work to derive the actual bit generation rate.

\begin{figure}
    \centering
    \includegraphics[width=.48\linewidth]{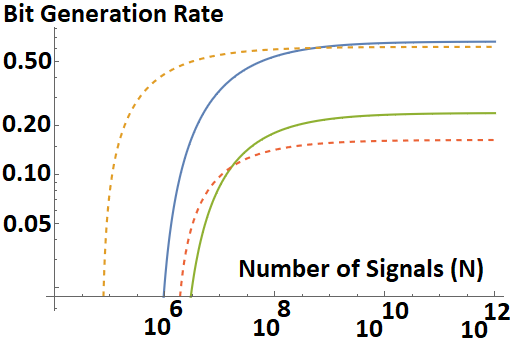}
    \includegraphics[width=.48\linewidth]{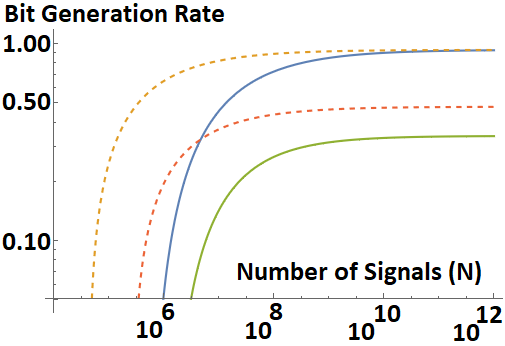}
    \caption{Evaluating and comparing the QRNG bit generation rates with biased measurements using our new result (Solid lines, Equation \ref{eq:QRNG-ours}) and prior work using standard entropic uncertainty (Dashed lines, Equation \ref{eq:QRNG-other}) as the number of signals $N$ (the $x$-axes) increases.  Left: Assuming $5\%$ noise (thus, $w(q) = .05$), Right: Assuming no noise ($w(q) = 0$). Blue: Our new result with no bias; Yellow: prior work with no bias; Green: Our new result with $b = 0.2$; Red: Prior work with $b = 0.2$.  We note that when there is some noise, our result clearly produces higher bit generation rates in all comparable cases as long as the number of signals is large enough.  When there is no noise, our result produces lower rates.  As the number of signals increases, our result converges to prior work when $b=0$, but produces worse results when $b = 0.2$ in the no noise case (right); however in the noisy case (left), our result surpasses prior work as the number of signals increases.  See text for more discussion.}
    \label{fig:QRNG:1}
\end{figure}

\begin{figure}
    \centering
    \includegraphics[width=.48\linewidth]{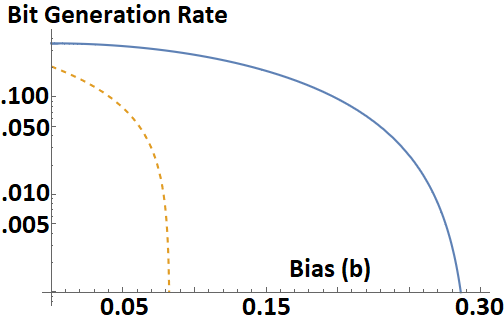}
    \includegraphics[width=.48\linewidth]{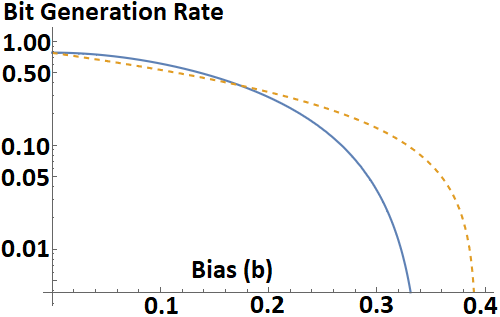}
    \caption{Evaluating and comparing the QRNG bit generation rates with biased measurements using our new result (Solid lines, Equation \ref{eq:QRNG-ours}) and prior work using standard entropic uncertainty (Dashed lines, Equation \ref{eq:QRNG-other}) as the bias parameter $b$ (the $x$-axes) increases. We fix $N = 10^{10}$ for these evaluations.  Left: Assuming a high level of noise at $15\%$ (thus, $w(q) = .15$), Right: Assuming a low level of noise at $2\%$ ($w(q) = 0.02$). Blue: Our new result; Yellow: prior work. Here, again, we see that at high noise our new rate produces substantially higher bit generation rates and has a higher tolerance to biased measurements, whereas at lower levels of noise, standard entropic uncertainty produces a better result in most cases (except for a low level of bias $b < .2$).}
    \label{fig:QRNG:2}
\end{figure}


\textbf{Quantum Key Distribution: } Next, we consider QKD.  Here, we derive a key-rate expression for standard BB84 \cite{QKD-BB84} where, however, instead of using the $Z$ and Hadamard bases as usual, Alice and Bob measure in either the $Z$ or the $\X$ basis.  Equivalently, Alice sends states in either the $Z$ or $\X$ basis while Bob measures in either basis.  Using results from \cite{tomamichel2012tight}, which depend on standard entropic uncertainty relations, the following key-rate for this protocol was derived:
\begin{equation}\label{eq:QKD-old}
r_{old} = \frac{1}{N}\left(n(c - h(w(q) + \mu)) - \leakEC - \log_2\frac{2}{\hat{\epsilon}^2}\right),
\end{equation}
where $\leakEC$ is the amount of information leaked during error correction, $c = -\log_2\left(\frac{1}{2}+b\right)$ and:
\begin{equation}
\mu = \sqrt{\frac{N(m+1)}{nm^2}\ln\frac{2}{\hat{\epsilon}}}.
\end{equation}
The above equations were derived using an entropic uncertainty relation from \cite{tomamichel2011uncertainty}.

On the other hand, our new relation in Theorem \ref{thm:main} can be used to immediately find the following key-rate for the protocol:
\begin{equation}\label{eq:QKD-new}
r_{new} = \frac{1}{N}\left(n(1 - \hat{h}(w(q) + \nu + \delta)) - \leakEC - \log_2\frac{1}{\epsilon}\right).
\end{equation}
Note we are ignoring an additional leakage of $\log\frac{1}{\epsilon_{cor}}$ \emph{in both key-rate expressions} caused by a final correctness check - however such a leakage would apply equally to both key-rate expressions and, since we are only interested in a direct comparison, this (small) leakage will not affect the results presented here.

As in the QRNG analysis, we set the failure probability and security level of both analysis methods to $10^{-12}$ which means setting $\hat{\epsilon} = 10^{-12}$ (for $r_{old}$) and setting $\epsilon = 10^{-36}$ for our new result ($r_{new}$).  We also use a sampling rate of $7\%$ (so $m = 0.07N$).  Finally, we use $\leakEC = 1.2h(w(q) + \delta)$ for our new work and $\leakEC = 1.2h(w(q) + \mu)$ for previous work ($r_{old}$); note that $\delta$ is usually larger than $\mu$ so this is actually to the advantage of prior work (as is setting $\epsilon$ so small to produce the same failure rate as prior work - this will actually benefit prior work in our comparison).  Despite this, our new result shows significant improvement over prior work in several, though not all, settings as shown in Figures \ref{fig:QKD1} and \ref{fig:QKD2}.  As shown in Figure \ref{fig:QKD1}, in the no noise and no bias case, prior work surpasses our work.  However, as the bias increases, our new bound surpasses prior work.  Figure \ref{fig:QKD2} again shows previous trends in that our bound is best when there is both bias and significant noise.  

\begin{figure}
    \centering
    \includegraphics[width=.48\linewidth]{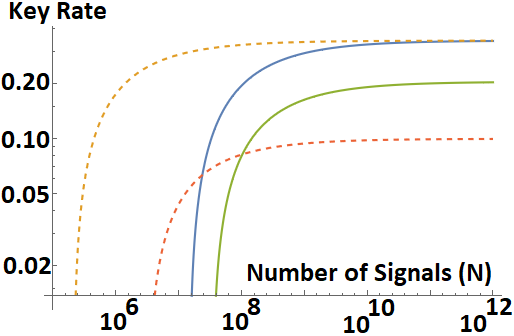}
    \includegraphics[width=.48\linewidth]{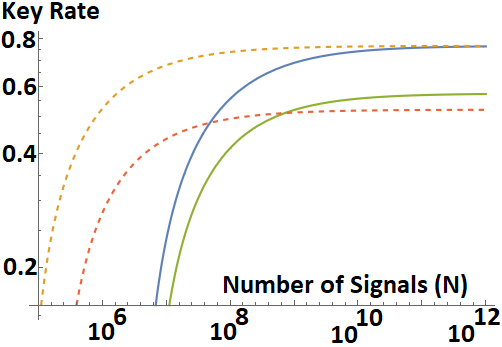}
    \caption{Evaluating and comparing the QKD key generation rates with biased measurements using our new result (Solid lines, Equation \ref{eq:QKD-new}) and prior work using standard entropic uncertainty (Dashed lines, Equation \ref{eq:QKD-old}) as the number of signals $N$ (the $x$-axes) increases.  Left: Assuming $5\%$ noise (thus, $w(q) = .05$), Right: Assuming $1\%$ noise ($w(q) = 0.01$). Blue Solid: Our new result with $b = 0$; Yellow Dashed: prior work with $b=0$; Green Solid: Our result with $b=0.1$; Red Dashed: Prior work with $b = 0.1$.}
    \label{fig:QKD1}
\end{figure}

\begin{figure}
    \centering
    \includegraphics[width=.48\linewidth]{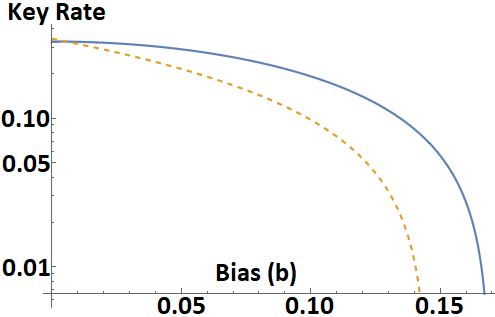}
    \includegraphics[width=.48\linewidth]{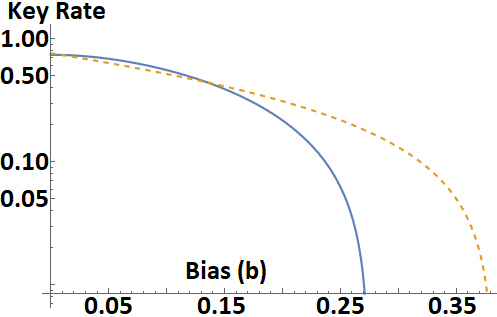}
    \caption{Evaluating and comparing the QKD key generation rates with biased measurements using our new result (Solid lines, Equation \ref{eq:QKD-new}) and prior work using standard entropic uncertainty (Dashed lines, Equation \ref{eq:QKD-old}) as the bias parameter $b$ (the $x$-axes) increases. We fix $N = 10^{10}$ for these evaluations.  Left: Assuming $5\%$ noise (thus, $w(q) = .05$), Right: Assuming a lower level of noise at $1\%$ ($w(q) = 0.01$). Blue solid: Our new result; Yellow dashed: prior work.}
    \label{fig:QKD2}
\end{figure}

\section{Closing Remarks}

In this work, we derived a new entropic uncertainty relation for biased measurements.  We applied our result to QRNG and QKD protocols and compared to prior work.  We also compared our relation in the asymptotic scenario to standard entropic uncertainty relations.  Our evaluations and comparisons showed that there are several cases where our new relation surpassed prior work, sometimes substantially so.  Our result seems best when there is both noise in the channel and bias in the measurements.  When the noise is very low or non-existent, prior work produced better results.  However, since our result, along with prior work, all produce lower-bounds on the min-entropy, users may simply evaluate both and take the maximum.  

Many interesting open questions remain.  Most important would be to remove the need for Definition \ref{def:sym}.  We suspect our method does not actually need this assumption on the state. It is only used in one part of the proof, to more easily bound the trace distance of two particular states, and we suspect other methods may be used for this.  Nonetheless, even with this assumption, our result is still highly practical to quantum cryptography.  Other open questions include extending this work to higher dimensions beyond qubits, and dealing with other imperfect measurements beyond bias.  We also did not compare to the generalized entropic uncertainty relation of \cite{tomamichel2013link} which may produce similar, or better results under biased settings.  Regardless, our result is specific to this particular instance of biased measurements and provides a bound that is easy to compute, whereas the generalized result of \cite{tomamichel2013link}, though very powerful and applicable to any scenario, requires an often difficult optimization to provide a good bound on the min entropy.  Our bound also handles all sampling effects automatically, which is an added benefit of the sampling-based approach used here.  We leave a further comparison between these two results as interesting future work.

\textbf{Acknowledgments:} The author would like to acknowledge support from the NSF under grant number 2143644.

\balance

\end{document}